\documentclass[times]{my-iapress}

\usepackage{moreverb}
\usepackage[colorlinks,bookmarksopen,bookmarksnumbered,citecolor=red,urlcolor=red]{hyperref}

\def\volumeyear{2019}
\def\volumenumber{7}
\def\volumemonth{June}
\usepackage{amsmath,amssymb}
\usepackage{graphicx}

\newtheorem{remark}{Remark}

\newtheorem{theorem}{Theorem}

\newtheorem{proposition}{Proposition}
\newtheorem{definition}{Definition}

\def\r{\mathbb{R}}

\allowdisplaybreaks

\begin{document}

\runningheads{A minimal HIV-AIDS infection model with general incidence rate}{%
E. M. Lotfi, M. Mahrouf, M. Maziane, C. J. Silva, D. F. M. Torres and N. Yousfi}

\title{A minimal HIV-AIDS infection model with general incidence rate\\
and application to Morocco data}

\author{El Mehdi Lotfi\affil{1},
Marouane Mahrouf\affil{1},
Mehdi Maziane\affil{1},
Cristiana J. Silva\affil{2}\corrauth,
Delfim F. M. Torres\affil{2},
Noura Yousfi\affil{1}}

\address{
\affilnum{1}Department of Mathematics and Computer Science,
Faculty of Sciences Ben M'sik, Hassan II University,
P.O Box 7955 Sidi Othman, Casablanca, Morocco.
\affilnum{2}Center for Research and Development in Mathematics and Applications (CIDMA),
Department of Mathematics, University of Aveiro, 3810-193 Aveiro, Portugal.}

\corraddr{Cristiana J. Silva (Email: cjoaosilva@ua.pt).
Department of Mathematics, University of Aveiro, 3810-193 Aveiro, Portugal.}


\begin{abstract}
We study the global dynamics of a SICA infection model with general
incidence rate. The proposed model is calibrated with cumulative cases
of infection by HIV--AIDS in Morocco from 1986 to 2015. We first prove
that our model is biologically and mathematically well-posed.
Stability analysis of different steady states is performed
and threshold parameters are identified where the model exhibits
clearance of infection or maintenance of a chronic infection.
Furthermore, we examine the robustness of the model to some parameter
values by examining the sensitivity of the basic reproduction number.
Finally, using numerical simulations with real data from Morocco,
we show that the model predicts well such reality.
\end{abstract}

\keywords{SICA compartmental model, general incidence function,
global stability, Lyapunov functionals.}

\maketitle

\noindent{\bf AMS 2010 subject classifications} 34D23, 92D30.


\section{Introduction}

Human immunodeficiency virus or HIV, is a type of virus that can cause
a disease called acquired immunodeficiency syndrome, or AIDS. HIV infection
affects the immune system, which is body's natural defence against diseases.
If left untreated, serious illness can occur. Normally innocuous infections,
such as influenza or bronchitis, can then get worse, become very difficult
to treat, or even cause death. In addition, the risk of cancer is also increased.
For this reason, the World Health Organization is committed
to end the AIDS epidemic by 2030 \cite{UNAIDS}.

According to the latest statistics on the state of AIDS epidemic by UNAIDS \cite{UNAIDS},
36.9 million people, globally, were living with HIV in 2017, of which 21.7 million
individuals were accessing ART (antiretroviral therapy) treatment; and 1.8 million became
newly infected with HIV in 2017. A total of 77.3 million individuals have become
infected with HIV since the start of the epidemic in 1981. Figures of death indicate
that 940,000 people died of AIDS-related illnesses in 2017, with a total
of 35.4 million people that have died from AIDS-related
illnesses since the start of the epidemic.

The risk of HIV infection depends on the mode of transmission, the prevalence
of HIV and the behaviour of the population. Human populations exhibit highly
variable behaviours in terms of sex and injecting drug use, a heterogeneity
of risk conceptualized by convention across three population groups
(see Figure~\ref{populationgroup}). The key group is those most exposed
to HIV infection, generally injecting drug users, men who have sex
with men (including sex workers) and sex workers. The second group
is the gateway population (e.g., truck drivers and clients of sex workers)
who pose a risk of intermediate exposure and puts the high-risk group
in contact with the third low-risk group, general population,
which brings together the bulk of the population of all communities.
\begin{figure}[!ht]
\includegraphics[scale=0.7]{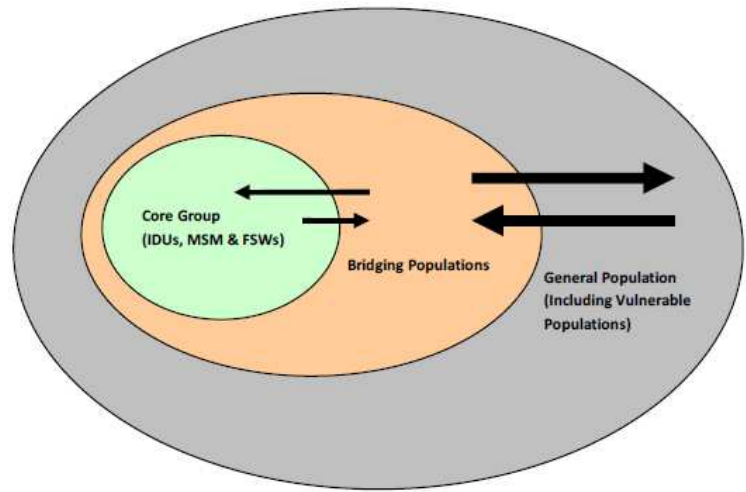}
\caption{Heterogeneity of the exposed risk to a sexually
transmitted infection in a given population.}
\label{populationgroup}
\end{figure}

Morocco, officially the Kingdom of Morocco, is a unitary sovereign
state located in North-western Africa that has a population of
34 996 167. The political regime of Morocco is a constitutional monarchy.
Moroccan culture is a blend of Berber, Arab, West African and European influences.
Its geographical proximity to Europe and its diverse cultures make it among
the most visited destinations with a ranking of 33rd out of 227 countries
\cite{populationdata}. In 2016, Morocco had more than 1000 new HIV infections
and almost 1000 AIDS-related deaths. There were 22 000 people living with HIV
in 2016, among whom 48\% were accessing ART treatment. In 2018, Morocco
began its thirtieth year of response to HIV, with great satisfaction.
Concerted efforts in the country have resulted in a 42\%
decline in new HIV infections since 2010, well above the only
4\% decrease in the Middle East and North Africa.
Coverage of HIV treatment has also increased in the country,
from 16\% in 2010 to 48\% in 2016 \cite{UNAIDS1}.

In our study, we assume that HIV infected individuals, with and without symptoms,
have access to ART treatment. The part of HIV infected individuals without AIDS
symptoms move to the \emph{chronic stage} (HIV-infected individuals
under ART treatment with a low viral load) \cite{MR3392642}.
In this class, when they respect carefully the ART treatment
and do not have a risky behaviour for HIV transmission,
individuals have the same life expectancy as uninfected ones
and the risk of HIV transmission is greatly reduced \cite{MyID:406,MR3392642,SICA}.
If the infected individuals do not take ART treatment, then
they can transmit the infection and the symptoms will start to appear.
We consider that these individuals are aware of their state of health,
so they take precautions not to infect others, and/or are too sick
to have a risky behaviour. Mathematically, this translates to the
fact that function $f$ of model \eqref{1} depends only on $S$ and $I$.
These reductionist assumptions have been
made to reduce the epidemic problem to a mathematical tractable form.
In order to contribute to the project of ending this epidemic by 2030,
we need to predict the evolution of its behaviour, which is a highly challenging problem.
One of the main tools to reach our goal is HIV infection mathematical modelling.
For this, based on \cite{SICA}, we propose the following model:
\begin{equation}
\label{1}
\begin{cases}
\dot{S}(t) = \Lambda  - \mu S(t)- f\left(S(t),I(t)\right)I(t),\\[0.2 cm]
\dot{I}(t) = f\left(S(t),I(t)\right)I(t) - (\rho + \phi + \mu)I(t)
+ \alpha A(t)  + \omega C(t), \\[0.2 cm]
\dot{C}(t) = \phi I(t) - (\omega + \mu)C(t),\\[0.2 cm]
\dot{A}(t) =  \rho \, I(t) - (\alpha + \mu + d) A(t),
\end{cases}
\end{equation}
with initial conditions
\begin{equation}
\label{IC}
S(0)=S_0\geq0, \
I(0)=I_0\geq0, \
C(0)=C_0\geq0, \
A(0)=A_0\geq 0.
\end{equation}
The model considers a varying population size in a homogeneously mixing population
and subdivides the human into four mutually-exclusive compartments:
\begin{itemize}
\item[-] Susceptible individuals ($S$);

\item[-] HIV infected individuals with no clinical symptoms of AIDS
(the virus is living or developing in the host
but without producing symptoms or only mild ones),
but that are able to transmit HIV to other individuals ($I$);

\item[-] HIV infected individuals under ART treatment
with a viral load remaining low ($C$);

\item[-] HIV infected individuals with AIDS clinical symptoms ($A$).
\end{itemize}
The susceptible population is increased by recruitment of susceptible
individuals at a rate $\Lambda$, while $\mu$ is the natural death
rate of all individuals. The transmission process is governed by
a general incidence function $f(S,I)$ \cite{MR3815138}.
Individuals in the class $I$ progress to the class $C$
at a rate $\phi$ and to the class $A$ at a rate $\rho$.
Individuals in the class $A$, that start taking ART treatment,
progress to the class $C$ at rate $\alpha$. Individuals in the class $C$
return to the class $I$, if the treatment is not maintained, at a rate $\omega$.
Individuals suffer from AIDS induced death at a rate $d$.
The total population at time $t$, denoted by $N(t)$, is given by
\begin{equation*}
N(t) = S(t) + I(t) + C(t) + A(t).
\end{equation*}
As in \cite{Hattaf1,HattafC,HattafD1,HattafD2,Lotfi1}, the incidence
function $f(S,I)$ is assumed to be non-negative and continuously differentiable
in the interior of $\r^{2}_{+}$ and satisfies the following
hypotheses:
\begin{gather}
\label{H1}\tag{$H_{1}$} f(0,I)=0,
\hspace*{0.2 cm} \text{ for all } I \geq 0,\\
\label{H2}\tag{$H_{2}$} \frac{ \partial f}{\partial S}(S,I)> 0,
\hspace*{0.2 cm} \text{ for all } S>0\ \text{ and } \ I \geq 0,\\
\label{H3}\tag{$H_{3}$} \frac{ \partial f}{\partial I}(S,I)
\leq 0 \hspace*{0.2 cm},\hspace*{0.2 cm} \text{ for all }\ S\geq 0\
\text{ and } \ I \geq 0.
\end{gather}
The reason for adopting hypothesis \eqref{H3} is the fact that susceptible individuals
take measures to reduce contagion if the epidemics breaks out.
This idea has first been explored in \cite{MR0529097}.

The rest of the paper is structured as follows. The mathematical analysis
of the model is presented in Section~\ref{sec:02}. Then we establish the
stability of the equilibria in Section~\ref{sec:03}. A sensitivity
analysis of the basic reproduction number is given in Section~\ref{sec:04}.
Finally, in Section~\ref{sec:05} an application of the model to HIV/AIDS
infection data from the period of 1986--2015 in Morocco is provided,
and some conclusions are derived in Section~\ref{sec:06}.


\section{Basic results}
\label{sec:02}

In this section, we prove that all solutions with nonnegative
initial data will remain nonnegative and bounded for all time.

\begin{theorem}
All solutions of \eqref{1} starting from non-negative initial conditions
\eqref{IC} exist for all $t>0$ and remain bounded and non-negative.
Moreover,
$$
N(t)\leq N(0)+\dfrac{\Lambda}{\mu}.
$$
\end{theorem}

\begin{proof}
For positivity, we show that any solution starting in the non-negative set
$\r^{4}_{+}=\{(S,I,C,A)\in\r^{4}: S\geqslant0; \ I\geqslant0;
\ C\geqslant0; \ A\geqslant0\}$ remains there forever:
$\big(S(t),I(t),C(t),A(t)\big)\in\r^{4}_{+}$. We have
\begin{eqnarray*}
\dot{S}\mid_{S=0} &=& \Lambda\geq0, \\
\dot{I}\mid_{I=0} &=& \omega C+\alpha A\geq0, \\
\dot{C}\mid_{C=0} &=& \phi I\geq0, \\
\dot{A}\mid_{A=0} &=& \rho I\geq0.
\end{eqnarray*}
Hence, positivity of all solutions initiating in $\r^{4}_{+}$ is guaranteed.
Now, we prove that the solutions are bounded.
Adding the four equations of system \eqref{1}, we get
\begin{eqnarray*}
\dot{N}(t) &=& \Lambda - \mu N(t) -d A(t), \\
\dot{N}(t)&\leq&  \Lambda - \mu N(t),
\end{eqnarray*}
and we deduce that
$$
N(t)\leq N(0)e^{-\mu t}+\dfrac{\Lambda}{\mu}(1-e^{-\mu t}),
$$
since $0\leq e^{-\mu t}\leq1$ and $1-e^{-\mu t}\leq1$. Thus,
$$
N(t)\leq N(0)+\dfrac{\Lambda}{\mu}.
$$
The proof is complete.
\end{proof}

It is easy to see that system \eqref{1} has a disease-free equilibrium
$E_f=(\Lambda/\mu,0,0,0)$. Note that the trivial equilibrium with $(0,0,0,0)$
is not possible by hypothesis \eqref{H1}. Following the approach presented by
van den Driessche and Watmough in \cite{Rzero}, the basic reproduction
number $R_0$ is obtained by calculating the spectral radius of the
next-generation matrix for system \eqref{1}. We get
\begin{equation}
\label{R0}
R_0=\dfrac{f(\Lambda/\mu,0)\xi_2\xi_3}{\mathcal{D}},
\end{equation}
where along all manuscript
\begin{equation}
\label{xis}
\begin{split}
\xi_1&=\rho+\phi+\mu,\\
\xi_2&=\omega+\mu,\\ 
\xi_3&=\alpha+\mu+d,
\end{split}
\end{equation}
and 
\begin{equation}
\label{mathcal:D}
\mathcal{D}=\mu[\xi_2(\xi_3+\rho)+\phi \xi_3+\rho d]+\rho \omega d.
\end{equation}
Biologically, the basic reproduction number represents the average
of new infected individuals produced by a single HIV-infected/AIDS
individual on contact in a completely susceptible population.

\begin{theorem}
(i) If $R_0\leq1$, then the system \eqref{1} has a unique
disease-free equilibrium of the form $E_f=(\Lambda/\mu,0,0,0)$.\\
(ii) If $R_0>1$, then the disease-free equilibrium is still
present and the system \eqref{1} has a unique endemic equilibrium
of the form $E^*(S^*,I^*,C^*,A^*)$, with
$S^*\in\bigg(0,\dfrac{\Lambda}{\mu}\bigg)$, $I^*>0$, $C^*>0$ and $A^*>0$.
\end{theorem}

\begin{proof}
Obviously, $E_f$ is the unique steady state of system \eqref{1}
when $R_0\leq1$. To find the other equilibrium of system \eqref{1},
we solve the system
\begin{equation}
\label{eq04}
\begin{split}
\Lambda  - \mu S- f\left(S,I\right)I &= 0,\\
f\left(S,I\right)I - \xi_1 I
+ \alpha A  + \omega C &= 0, \\
\phi I - \xi_2C &= 0,\\
\rho \, I - \xi_3 A &= 0,
\end{split}
\end{equation}
where $\xi_i$, $i = 1, 2, 3$, are given as in \eqref{xis}. 
From \eqref{R0} and \eqref{eq04}, we get
$$
C=\dfrac{\phi}{\xi_2}I \qquad\text{  and }\qquad A=\dfrac{\rho }{\xi_3}I.
$$
Adding \eqref{1} and \eqref{IC}, we obtain
$$
\Lambda  - \mu S-\bigg(\xi_1-\dfrac{\omega \phi}{\xi_2}
-\dfrac{\alpha \rho}{\xi_3}\bigg)I=0.
$$
Then,
$$
I=\dfrac{(\Lambda-\mu S)R_0}{f\big(\Lambda/\mu,0\big)}.
$$
Now, we consider the following function defined on the interval
$\bigg[0,\dfrac{\Lambda}{\mu}\bigg]$ by
$$
G(S)=f\bigg(S,\dfrac{(\Lambda-\mu S)R_0}{f\big(\Lambda/\mu,0\big)}\bigg)
-\xi_1+\dfrac{\omega \phi}{\xi_2}+\dfrac{\alpha \rho}{\xi_3}.
$$
By using the intermediate value theorem on the function $G(S)$,
we show that the second equation of system \eqref{eq04} admits a
unique solution. Clearly, $G(0)=-\dfrac{\mathcal{D}}{\xi_2 \xi_3}<0$
and $G(\Lambda/\mu)=\dfrac{\mathcal{D}}{\xi_2 \xi_3}\big[R_0-1\big]$, 
with $\mathcal{D}$ as in \eqref{mathcal:D}, and
$$
G'(S)=\dfrac{\partial f}{\partial S}
-\dfrac{\mu R_0}{f\big(\Lambda/\mu,0\big)}\dfrac{\partial f}{\partial I}>0.
$$
By hypotheses \eqref{H2} and \eqref{H3}, if $R_0>1$,
then the system admits a unique endemic equilibrium
$E^*(S^*,I^*,C^*,A^*)$ with $S^*\in\left(0,\dfrac{\Lambda}{\mu}\right)$,
$I^*>0$, $C^*>0$ and $A^*>0$. This completes the proof.
\end{proof}


\section{Stability analysis}
\label{sec:03}

In this section, we focus on the stability analysis of the equilibria.
In the case of $R_0\leq1$, we have:

\begin{theorem}
\label{Ef}
The disease-free equilibrium $E_f$
is globally asymptotically stable if
$R_0\leq1$.
\end{theorem}

\begin{proof}
To study the global stability of $E_f$ for system \eqref{1},
we construct a Lyapunov functional at $E_f$ as follows:
\begin{equation*}
V_{1}(S,I,C,A)=S-S_{0}-\int_{S_{0}}^{S}\frac{f(S_{0},0)}{f(X,0)}dX
+I+\dfrac{\omega}{\xi_2}C+\dfrac{\alpha}{\xi_3}A,
\end{equation*}
where $S_0=\dfrac{\Lambda}{\mu}$ and $\xi_i$, $i = 2, 3$, 
are given as in \eqref{xis}. It is clear from hypothesis 
\eqref{H2} that $V_{1}$ is non-negative.
The time derivative of $V_{1}$ along
the solution of system \eqref{1} is given by
\begin{eqnarray*}
\dfrac{d V_1}{dt}
&=& \left(1-\dfrac{f(S_{0},0)}{f(S,0)}\right)\dot{S}
+\dot{I}+\dfrac{\omega}{\xi_2}\dot{C}+\dfrac{\alpha}{\xi_3}\dot{A}\\
&=& \mu \left(1-\dfrac{f(S_{0},0)}{f(S,0)}\right)(S_0-S)
+I\left[\dfrac{f(S,I)}{f(S,0)}f(S_0,0)-\left(\xi_1
-\dfrac{\omega \phi}{\xi_2}-\dfrac{\alpha \rho}{\xi_3}\right)\right]\\
&=&\mu \left(1-\dfrac{f(S_{0},0)}{f(S,0)}\right)(S_0-S)
+\dfrac{\mathcal{D}}{\xi_2 \xi_3}I\left(\dfrac{f(S,I)}{f(S,0)}R_0-1\right)\\
&\leq& \mu \left(1-\dfrac{f(S_{0},0)}{f(S,0)}\right)(S_0-S)
+\dfrac{\mathcal{D}}{\xi_2 \xi_3}I\left(R_0-1\right),
\end{eqnarray*}
where $\mathcal{D}$ is given in \eqref{mathcal:D}
and $\xi_i$, $i = 1, 2, 3$, are defined by \eqref{xis}.
Using the inequalities
\begin{eqnarray*}
1-\dfrac{f(S_{0},0)}{f(S,0)}
&\geq& 0 \quad \text{for}\quad S\geq S_0,\\
1-\dfrac{f(S_{0},0)}{f(S,0)}
&<& 0 \quad \text{for}\quad S<S_0,
\end{eqnarray*}
we have
\begin{equation*}
\left(1-\dfrac{f(S_{0},0)}{f(S,0)}\right)(S_0-S)\leq 0.
\end{equation*}
Therefore, $ \frac{d V_1}{dt}\leq0$. Furthermore, the largest compact
invariant set in $\left\{(S,I,C,A)\mid\frac{d V_1}{dt}=0\right\}$
is just the singleton $E_f$. Using LaSalle's invariance principle \cite{LaSalle},
we conclude that $E_f$ is globally asymptotically stable.
\end{proof}

Now, we study the stability of the equilibria in the case $R_0>1$.
For this, we assume that function $f$ satisfies the following condition:
\begin{gather}
\label{H4}\tag{$H_{4}$}
\bigg(1-\dfrac{f(S,I)}{f(S,I^*)}\bigg)\bigg(\dfrac{f(S,I^*)}{f(S,I)}
-\dfrac{I}{I^*}\bigg)\leq0, \ \text{ for all } \ S,I>0.
\end{gather}

\begin{theorem}
(i) If $R_0>1$, then the disease-free equilibrium $E_f$ is unstable.\\
(ii) If $R_0>1$ and \eqref{H4} holds, then the endemic equilibrium
$E^*$ is globally asymptotically stable.
\end{theorem}

\begin{proof}
We begin by proving the instability of $E_f$. For any arbitrary equilibrium $\overline{E}(\overline{S},\overline{I},\overline{C},\overline{A})$,
the characteristic equation is given by
$$
\begin{array}{ccc}
\left|
\begin{array}{cccc}
-\mu-\dfrac{\partial f}{\partial S}\overline{I}-\lambda
& -\dfrac{\partial f}{\partial I}\overline{I}-f(\overline{S},\overline{I})
& 0 & 0 \\ [8 pt]
\dfrac{\partial f}{\partial S}\overline{I}
& \dfrac{\partial f}{\partial I}\overline{I}
+f(\overline{S},\overline{I})-\xi_1 -\lambda
& \omega & \alpha \\[8 pt]
0 & \phi & -(\xi_2 +\lambda) & 0 \\[8 pt]
0 & \rho & 0 & -(\xi_3 +\lambda)
\end{array}
\right| & = & 0,
\end{array}
$$
where $\xi_i$, $i = 1, 2, 3$, are defined in \eqref{xis}.
Evaluating the characteristic equation at $E_f(S_0,0,0,0)$, we have
$$
\lambda^3+a_1\lambda^2+a_2\lambda+a_3=0,
$$
where
\begin{eqnarray*}
a_1 &=& \xi_1+\xi_2+\xi_3- f(S_0,0),\\
a_2 &=& \xi_1\xi_3+\xi_1\xi_2+\xi_2\xi_3
-(\xi_2+\xi_3)f(S_0,0)-\phi \omega-\rho \alpha,\\
a_3 &=& (1-R_0)\mathcal{D},
\end{eqnarray*}
$\mathcal{D}$ as in \eqref{mathcal:D}.
It is clear that $a_3<0$ when $R_0>1$. Then,
the disease-free equilibrium $E_f$ is unstable.
Now, we establish the global stability of the endemic equilibrium $E^*$.
For that, we define a Lyapunov functional $V_2$ as follows:
\begin{multline*}
V_{2}(S,I,C,A)
=S-S^{*}-\displaystyle{\int_{S^{*}}^{S}}\frac{f(S^{*},I^{*})}{f(X,I^*)}dX
+I-I^{*}-I^{*}\ln\left(\dfrac{I}{I^*}\right)\\
+\dfrac{\omega}{\xi_2}\left(C-C^{*}-C^{*}\ln\left(\dfrac{C}{C^*}\right)\right)
+\dfrac{\alpha}{\xi_3}\left(A-A^{*}-A^{*}\ln\left(\dfrac{A}{A^*}\right)\right).
\end{multline*}
The time derivative of $V_2$ along the positive solutions
of system \eqref{1} satisfies
\begin{eqnarray*}
\dfrac{d V_2}{dt}
&=& \left(1-\dfrac{f(S^{*},I^{*})}{f(S,I^*)}\right)\dot{S}
+\left(1-\dfrac{I^{*}}{I}\right)\dot{I}+\dfrac{\omega}{\xi_2}\left(
1-\dfrac{C^{*}}{C}\right)\dot{C}+\dfrac{\alpha}{\xi_3}\left(
1-\dfrac{A^{*}}{A}\right)\dot{A}.
\end{eqnarray*}
By applying $\Lambda=\mu S^* +f(S^{*},I^{*})I^{*}$ and
$\xi_1 I^*=f(S^{*},I^{*})I^{*}+\omega C^* +\alpha A^*$, we get
\begin{eqnarray*}
\dfrac{d V_2}{dt}
&=& \mu(S^{*}-S) \left(1-\dfrac{f(S^{*},I^{*})}{f(S,I^*)}\right)
+f\left(S^{*},I^{*}\right)I^{*}\left(1-\dfrac{f(S^{*},I^{*})}{f(S,I^*)}\right)\\
& & -f(S,I)I(1-\dfrac{f(S^*,I^*)}{f(S,I^*)})+f(S,I)I-\xi_1I+\alpha A+\omega C \\
& & -f(S,I)I^*+\xi_1 I^*-\dfrac{\alpha AI^*}{I}
-\dfrac{\omega CI^*}{I}+\dfrac{\omega \phi I}{\xi_2}\\
& &-\omega C-\dfrac{\omega \phi C^* I}{\xi_2 C} +\omega C^*
+\dfrac{\alpha \rho I}{\xi_3}-\alpha A 
-\dfrac{\alpha \rho A^* I}{\xi_3 A}+\alpha A^*\\
&=& \mu(S^{*}-S) \left(1-\dfrac{f(S^{*},I^{*})}{f(S,I^*)}\right)
+f(S^{*},I^{*})I^{*}\left(1-\dfrac{f(S^{*},I^{*})}{f(S,I^*)}\right)\\
& & +\dfrac{f(S^*,I^*)f(S,I)I}{f(S,I^*)}+\omega C^*\left(
1-\dfrac{CI^*}{C^*I}\right)+\alpha A^*\left(1-\dfrac{AI^*}{A^*I}\right)\\
& &+\dfrac{\omega \phi}{\xi_2}I^*\left(1-\dfrac{C^*I}{CI^*}\right)
-\dfrac{\omega \phi I^*}{\xi_2}+\dfrac{\alpha \rho}{\xi_3}I^*\left(1
-\dfrac{A^*I}{AI^*}\right)-\dfrac{\alpha \rho I^*}{\xi_3}-\xi_1I-f(S,I)I^*+\xi_1I^*\\
& & +\dfrac{\omega \phi I}{\xi_2}+\dfrac{\alpha \rho I}{\xi_3}+f(S^*,I^*)I^*-f(S^*,I^*)I^*\\
&=& \mu(S^{*}-S) \left(1-\dfrac{f(S^{*},I^{*})}{f(S,I^*)}\right) 
+2f(S^*,I^*)I^* -\dfrac{f(S^*,I^*)^2I^*}{f(S,I^*)}\\
& & +\dfrac{f(S^*,I^*)f(S,I)I}{f(S,I^*)}+\omega C^*\left(2
-\dfrac{CI^*}{C^*I}-\dfrac{C^*I}{CI^*}\right)\\
& &+\alpha A^*\left(2-\dfrac{AI^*}{A^*I}-\dfrac{A^*I}{AI^*}\right)
+\dfrac{I}{I^*}\left(f(S^*,I^*)I^*-\xi_1I^*+\alpha A^*+\omega C^*\right)-f(S^*,I^*)I.
\end{eqnarray*}
By $f(S^*,I^*)I^*-\xi_1I^*+\alpha A^*+\omega C^*=0$, we get
\begin{eqnarray*}
\dfrac{d V_2}{dt}&=& \mu(S^{*}-S) \left(1-\dfrac{f(S^{*},I^{*})}{f(S,I^*)}\right)
+\omega C^*\left(2-\dfrac{CI^*}{C^*I}-\dfrac{C^*I}{CI^*}\right)\\
& &+\alpha A^*\left(2-\dfrac{AI^*}{A^*I}-\dfrac{A^*I}{AI^*}\right)
+2f(S^*,I^*)I^*+\dfrac{f(S^*,I^*)f(S,I)I}{f(S,I^*)}\\
& &-\dfrac{f(S^*,I^*)^2I^*}{f(S,I^*)}-f(S^*,I^*)I-f(S,I)I^*\\
&=& \mu(S^{*}-S) \left(1-\dfrac{f(S^{*},I^{*})}{f(S,I^*)}\right)
+\omega C^*\left(2-\dfrac{CI^*}{C^*I}-\dfrac{C^*I}{CI^*}\right)\\
& &+\alpha A^*\left(2-\dfrac{AI^*}{A^*I}-\dfrac{A^*I}{AI^*}\right)
+3f(S^*,I^*)I^*-f(S^*,I^*)I^*+\dfrac{f(S^*,I^*)f(S,I)I}{f(S,I^*)}\\
& &-\dfrac{f(S^*,I^*)^2I^*}{f(S,I^*)}-f(S^*,I^*)I-f(S,I)I^*
+\dfrac{f(S,I^*)f(S^*,I^*)I^*}{f(S,I)}-\dfrac{f(S,I^*)f(S^*,I^*)I^*}{f(S,I)}\\
&=& \mu(S^{*}-S) \left(1-\dfrac{f(S^{*},I^{*})}{f(S,I^*)}\right)
+\omega C^*\left(2-\dfrac{CI^*}{C^*I}-\dfrac{C^*I}{CI^*}\right)\\
& &+\alpha A^*\left(2-\dfrac{AI^*}{A^*I}-\dfrac{A^*I}{AI^*}\right)
+f(S^*,I^*)I^*\left[-1+\dfrac{f(S,I)I}{f(S,I^*)I^*}-\dfrac{I}{I^*}
+\dfrac{f(S,I^*)}{f(S,I)}\right]\\
& & +f(S^*,I^*)I^*\left[3-\dfrac{f(S^*,I^*)}{f(S,I^*)}
-\dfrac{f(S,I)}{f(S^*,I^*)}-\dfrac{f(S,I^*)}{f(S,I)}\right].
\end{eqnarray*}
Furthermore, we can easily see that 
$$
-1+\dfrac{f(S,I)I}{f(S,I^*)I^*}-\dfrac{I}{I^*}
+\dfrac{f(S,I^*)}{f(S,I)}
=\left(1-\dfrac{f(S,I)}{f(S,I^*)}\right)\left(
\dfrac{f(S,I^*)}{f(S,I)}-\dfrac{I}{I^*}\right).
$$
Thus,
\begin{equation}
\label{eq:rv2:1}
\begin{split}
\dfrac{d V_2}{dt}
&= \mu(S^{*}-S) \left(1-\dfrac{f(S^{*},I^{*})}{f(S,I^*)}\right)
+\omega C^*\left(2-\dfrac{CI^*}{C^*I}-\dfrac{C^*I}{CI^*}\right)\\
& \quad +\alpha A^*\left(2-\dfrac{AI^*}{A^*I}-\dfrac{A^*I}{AI^*}\right)
+f(S^*,I^*)I^*\left(1-\dfrac{f(S,I)}{f(S,I^*)}\right)\left(
\dfrac{f(S,I^*)}{f(S,I)}-\dfrac{I}{I^*}\right)\\
& \quad +f(S^*,I^*)I^*\left[3-\dfrac{f(S^*,I^*)}{f(S,I^*)}
-\dfrac{f(S,I)}{f(S^*,I^*)}-\dfrac{f(S,I^*)}{f(S,I)}\right].
\end{split}
\end{equation}
Since the arithmetic mean is greater than or equal
to the geometric mean, it is clear that
\begin{equation}
\label{eq:rv2:2}
\begin{array}{l}
2-\dfrac{CI^*}{C^* I}-\dfrac{C^*I}{CI^* }\leq 0, \\[12 pt]
2-\dfrac{AI^*}{A^* I}-\dfrac{A^*I}{AI^* }\leq 0,\\[12 pt]
3-\dfrac{f(S^*,I^*)}{f(S,I^*)}-\dfrac{f(S,I)}{f(S^*,I^*)}
-\dfrac{f(S,I^*)}{f(S,I)}\leq 0,
\end{array}
\end{equation}
and the equalities hold only for $S=S^*$, $I=I^*$, $C=C^*$ and $A=A^*$.
Using the inequalities
\begin{eqnarray*}
1-\dfrac{f(S^*,I^*)}{f(S,I^*)} &\geq& 0 \quad \text{for}\quad S\geq S^*,\\
1-\dfrac{f(S^*,I^*)}{f(S,I^*)} &<& 0 \quad \text{for}\quad S<S^*,
\end{eqnarray*}
which are ensured by hypothesis \eqref{H2}, we have
\begin{equation}
\label{eq:rv2:3}
(S^*-S)\left(1-\dfrac{f(S^*,I^*)}{f(S,I^*)}\right)\leq 0.
\end{equation}
By \eqref{H4}, we have 
\begin{equation}
\label{eq:rv2:4}
\left(1-\dfrac{f(S,I)}{f(S,I^*)}\right)\left(
\dfrac{f(S,I^*)}{f(S,I)}-\dfrac{I}{I^*}\right)
\leq 0.
\end{equation}
Therefore, it follows from \eqref{eq:rv2:1}, \eqref{eq:rv2:2}, \eqref{eq:rv2:3}
and \eqref{eq:rv2:4} that $\frac{d V_2}{dt}\leq0$. Furthermore, the largest
compact invariant set in $\left\{(S,I,C,A)\mid\frac{d V_2}{dt}=0\right\}$
is just the singleton $E^*$. Using LaSalle's invariance principle \cite{LaSalle},
we conclude that $E^*$ is globally asymptotically stable.
\end{proof}


\section{Sensitivity of the basic reproduction number}
\label{sec:04}

Often, susceptible individuals acquire HIV infection by contact
with individuals living with the virus but with no clinical AIDS symptoms
or only mild ones. Thus, the incidence function $f(S,I)$
is assumed to depend on the \emph{effective contact rate $\beta>0$}.
Then, $f$ can take many forms. Table~\ref{Table:1} collects
the most popular of such forms that one can find in the existing literature.
For any form of $f(S,I)$ given in Table~\ref{Table:1},
it is easy to verify that $\dfrac{\partial f(S,I)}{\partial \beta}
=\dfrac{f(S,I)}{\beta}$, which is important for examining
the robustness of model \eqref{1} to $\beta$.
\begin{table}[!ht]
\centering
\caption{Some special incidence functions $f(S,I)$,
where $\alpha_i\geq0$, $i=0,\ldots,3$.}
\label{Table:1}
{\renewcommand{\arraystretch}{2}	
\begin{tabular}{l c c}
\hline\hline
Incidence functions &   ${f(S,I)}$ &  References \\
\hline\hline
Bilinear  & $\beta S$ & \cite{Bilinear1,Bilinear2,Bilinear} \\[0.15cm]
\hline
Saturated & $\dfrac{\beta S}{1+\alpha_1S}$ or $\dfrac{\beta S}{1+\alpha_2I}$
& \cite{Saturated1,Saturated2} \\[0.15cm]
\hline
Beddington--DeAngelis & $\dfrac{\beta S }{1+\alpha_1S+\alpha_2I}$
& \cite{Beddington,Cantrell,DeAngelis}\\[0.15cm]
\hline
Crowley--Martin & $\dfrac{\beta S}{1+\alpha_1S+\alpha_2I+\alpha_1 \alpha_2SI}$
& \cite{Crowley,Liu,Zhou}\\[0.15cm]
\hline
Specific nonlinear  & $\dfrac{\beta S}{1+\alpha_1S+\alpha_2I+\alpha_3SI}$
& \cite{mahrouf1,Hattaf2,HattafD2,Lotfi,Maziane1}\\[0.15cm]
\hline
Hattaf--Yousfi & $\dfrac{\beta S}{\alpha_0+\alpha_1S+\alpha_2I+\alpha_3SI}$
& \cite{Hattaf-Yousfi,mahrouf} \\[0.15cm]
\hline\hline
\end{tabular}}
\end{table}

To determine the robustness of model \eqref{1} to some parameter values,
including $\beta$, we examine the sensitivity of the basic
reproduction number $R_0$ with respect to these parameters
by the so called \emph{sensitivity index}.

\begin{definition}[See \cite{Chitnis,Rodriques}]
\label{sensi}
The normalized forward sensitivity index of a variable $u$,
that depends differentially on a parameter $p$, is defined as
\begin{equation}
\label{sensidf}
\Upsilon^{u}_p := \dfrac{\partial u}{\partial p}\times\dfrac{p}{u}.
\end{equation}
\end{definition}

From \eqref{R0} and Definition~\ref{sensi}, we derive the normalized
forward sensitivity index of $R_0$ with respect to $\beta$, using any
form of incidence functions cited above, and we get the following proposition.

\begin{proposition}
The normalized forward sensitivity index of $R_0$ with
respect to $\beta$ is given by
\begin{equation*}
\Upsilon^{R_0}_\beta =\dfrac{\partial f(S_0,0)}{\partial \beta}
\times\dfrac{\beta}{f(S_0,0)}.
\end{equation*}
\end{proposition}

\begin{proof}
It is a direct consequence of \eqref{R0} and \eqref{sensidf}.
\end{proof}

\begin{remark}
The sensitivity index of $R_0$ \eqref{R0} of the model
with respect to $\phi$, $\rho$, $\alpha$
and $\omega$ are given, respectively, by $\Upsilon^{R_0}_\phi
=-\dfrac{\mu \phi \xi_3}{\mathcal{D}}$, $\Upsilon^{R_0}_\alpha
=\dfrac{\rho \alpha(\mu + d)\xi_2}{\mathcal{D} \xi_3}$,
$\Upsilon^{R_0}_\rho =-\dfrac{\rho(\mu+d)\xi_2}{\mathcal{D}}$
and $\Upsilon^{R_0}_\omega =\dfrac{\mu \omega \phi \xi_3}{\mathcal{D}\xi_2}$,
$\xi_i$, $i = 2, 3$, as in \eqref{xis} and $\mathcal{D}$ as in \eqref{mathcal:D},
which are given in \cite{SICA} using the bilinear incidence function.
\end{remark}

\begin{remark}
For all incidence functions in Table~\ref{Table:1},
$\beta$ is always the most sensitive parameter and
have a high impact on $R_0$. Indeed, $\Upsilon^{R_0}_\beta$
is independent of any parameter of system \eqref{1}
with $\Upsilon^{R_0}_\beta =+1$.
\end{remark}


\section{Numerical simulations: application to Moroccan HIV/AIDS data}
\label{sec:05}

Taking into account the data from the Ministry of Health
in Morocco \cite{url:HIVdata:morocco},
we estimate the value of the HIV transmission rate to be
$\beta = 0.755$. Moreover, we consider the initial conditions
\eqref{eq:initcond} based on Moroccan data:
\begin{equation}
\label{eq:initcond}
S_0 = (N_0-(2+9))/N_0 ,\quad
\ I_0 = 2/N_0,\quad
\ C_0 = 0,\quad
\ A_0 = 9/N_0,
\end{equation}
with the initial total population $N_0 = 23023935$
\cite{url:worldbank:morocco}.
The values of the parameters $\rho = 0.1$ and $\alpha = 0.33$ are taken
from \cite{Sharomi:MathBio:2008} and \cite{Bhunu:BMB:2009:HIV:TB}, respectively.
We assume that after one year, the HIV infected individuals $I$ that are under
ART treatment have a low viral load \cite{Perelson} and, therefore,
are transferred to the class $C$. In agreement, we take $\phi=1$.
It is well known that taking ART therapy is a long-term commitment.
In our simulations, following \cite{SICA}, we assume that the default
treatment rate for $C$ individuals is approximately 11 years
($1/\omega$ years, to be precise).
Following \cite{url:popdata:morocco}, the natural death and recruitment
rates are assumed to take the values $\mu = 1/74.02$
and $\Lambda = 2.19 \mu$. The AIDS induced death
rate is assumed to be $d = 1$ based on \cite{ZwahlenEggerUNAIDS}.
All the considered parameter values are summarized in Table~\ref{values1}.
\begin{table}[!htb]
\centering
\caption{Parameters of the HIV/AIDS model \eqref{1} for Morocco data
between 1986 and 2015 \cite{url:HIVdata:morocco}.}
\label{values1}
\begin{tabular}{l  p{6.5cm} l l}
\hline \hline
{\small{Symbol}} &  {\small{Description}} & {\small{Value}} & {\small{References}}\\
\hline
{\small{$N(0)$}} & {\small{Initial population}} & {\small{$23023935$}}
& {\small{\cite{url:worldbank:morocco}}}\\
{\small{$\mu$}} & {\small{Natural death rate}} & {\small{$1/74.02$}}
& {\small{\cite{url:popdata:morocco}}}\\
{\small{$\Lambda$}} & {\small{Recruitment rate}} & {\small{$2.19\, \mu$}}
& {\small{\cite{url:popdata:morocco}}}\\
{\small{$\beta$}} & {\small{HIV transmission rate}} & {\small{$0.755$}} & {\small{Estimated}}\\
{\small{$\phi$}} & {\small{HIV treatment rate for $I$ individuals}} &  {\small{$1$}}
& {\small{\cite{Perelson}}} \\
{\small{$\rho$}} & {\small{Default treatment rate for $I$ individuals}}
& {\small{$0.1 $}} & {\small{\cite{Sharomi:MathBio:2008}}}\\
{\small{$\alpha$}} & {\small{AIDS treatment rate}}
& {\small{$0.33 $}} & {\small{\cite{Bhunu:BMB:2009:HIV:TB}}}\\
{\small{$\omega$}} & {\small{Default treatment rate for $C$ individuals}}
& {\small{$0.09$}} & {\small{\cite{SICA}}}\\
{\small{$d$}} & {\small{AIDS induced death rate}} & {\small{$1$}}
& {\small{\cite{ZwahlenEggerUNAIDS}}}\\
\hline \hline
\end{tabular}
\end{table}

In Figure~\ref{fig:model:fit}, we observe that model \eqref{1}
fits the real data reported in \cite{url:HIVdata:morocco}.
The HIV cases described by model \eqref{1} are given by
$I(t) + C(t) + \mu\left( I(t) + C(t)\right)$
for $t \in [0, 29]$, which corresponds to the interval of time between
the years of 1986 ($t=0$) and 2015 ($t=29$).
\begin{figure}[!ht]
\advance\leftskip1.6cm
\includegraphics[width=0.7\textwidth]{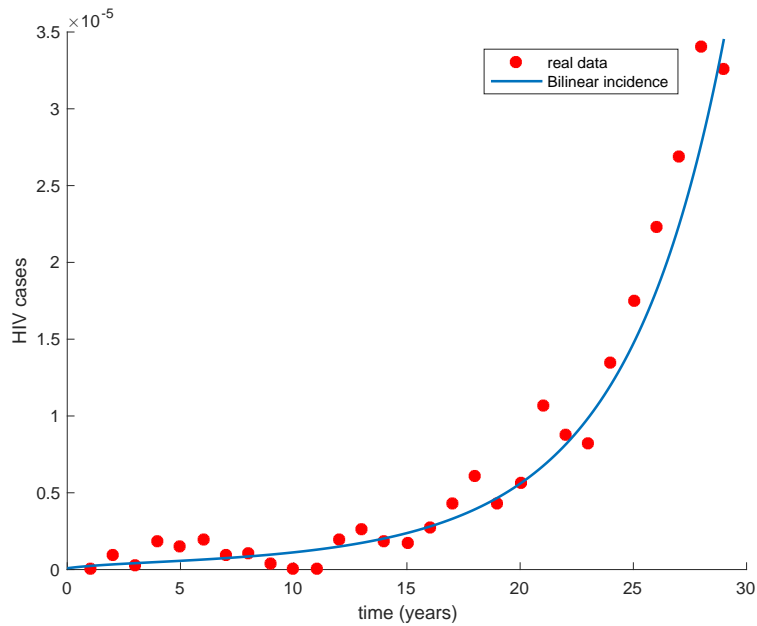}
\caption{Model \eqref{1} fitting the data
of HIV cases in Morocco between 1986 ($t = 0$) and 2015 ($t = 29$).}
\label{fig:model:fit}
\end{figure}

For the saturated incidence function  $\dfrac{\beta S}{1+\alpha_1S}$
\cite{Saturated1}, see Figure~\ref{fig:saturated:alpha1}.
\begin{figure}[!ht]
\advance\leftskip1.6cm
\includegraphics[width=0.7\textwidth]{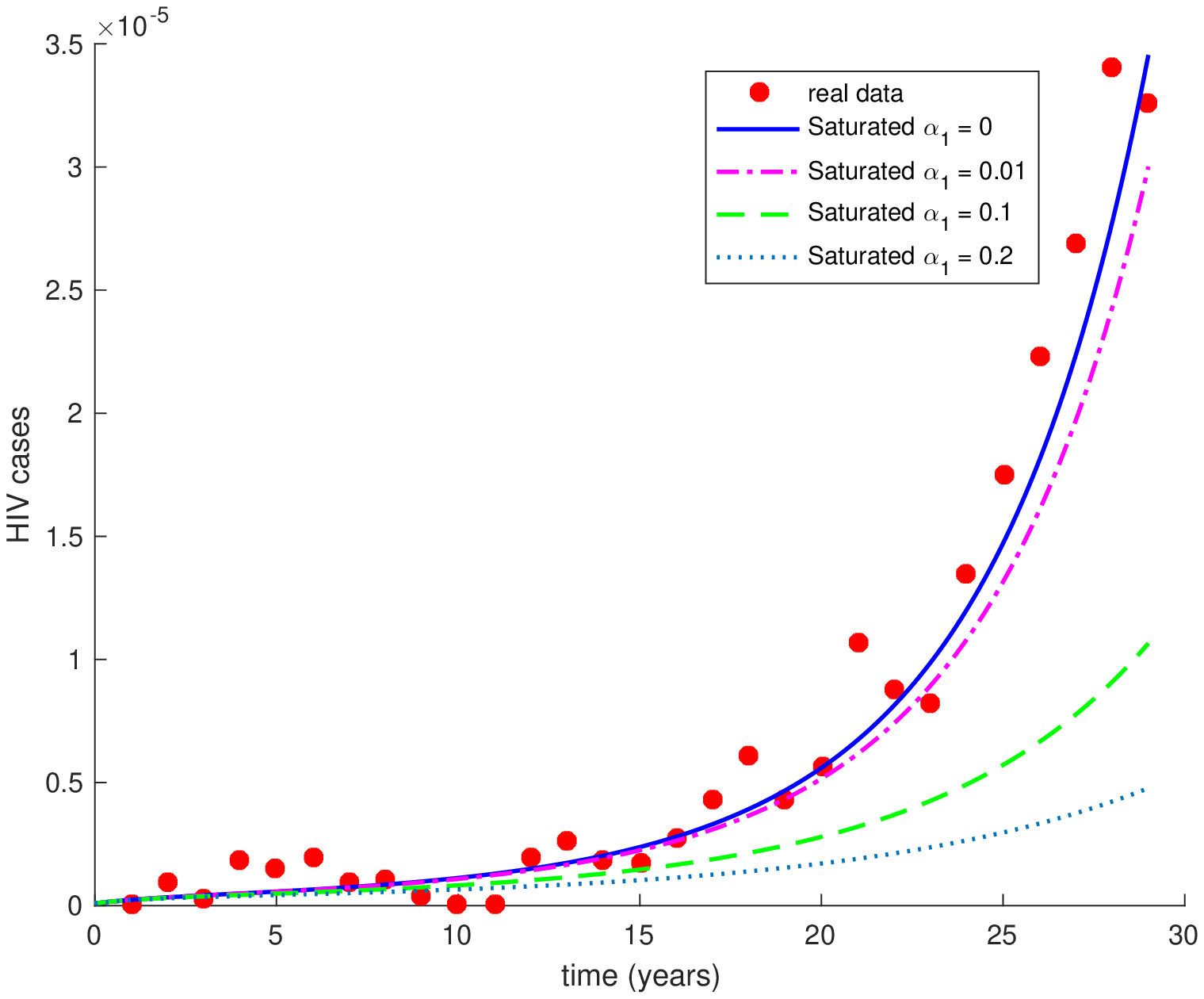}
\caption{Saturated incidence function $\dfrac{\beta S}{1+\alpha_1S}$
with $\alpha_1 \in \{0, 0.01, 0.1, 0.2  \}$.
Time $t = 0$ corresponds to the year 1986.}
\label{fig:saturated:alpha1}
\end{figure}

For the saturated incidence function $\dfrac{\beta S}{1+\alpha_2 I}$
\cite{Saturated2}, see Figure~\ref{fig:saturated:alpha2}.
\begin{figure}[!ht]
\advance\leftskip1.6cm
\includegraphics[width=0.7\textwidth]{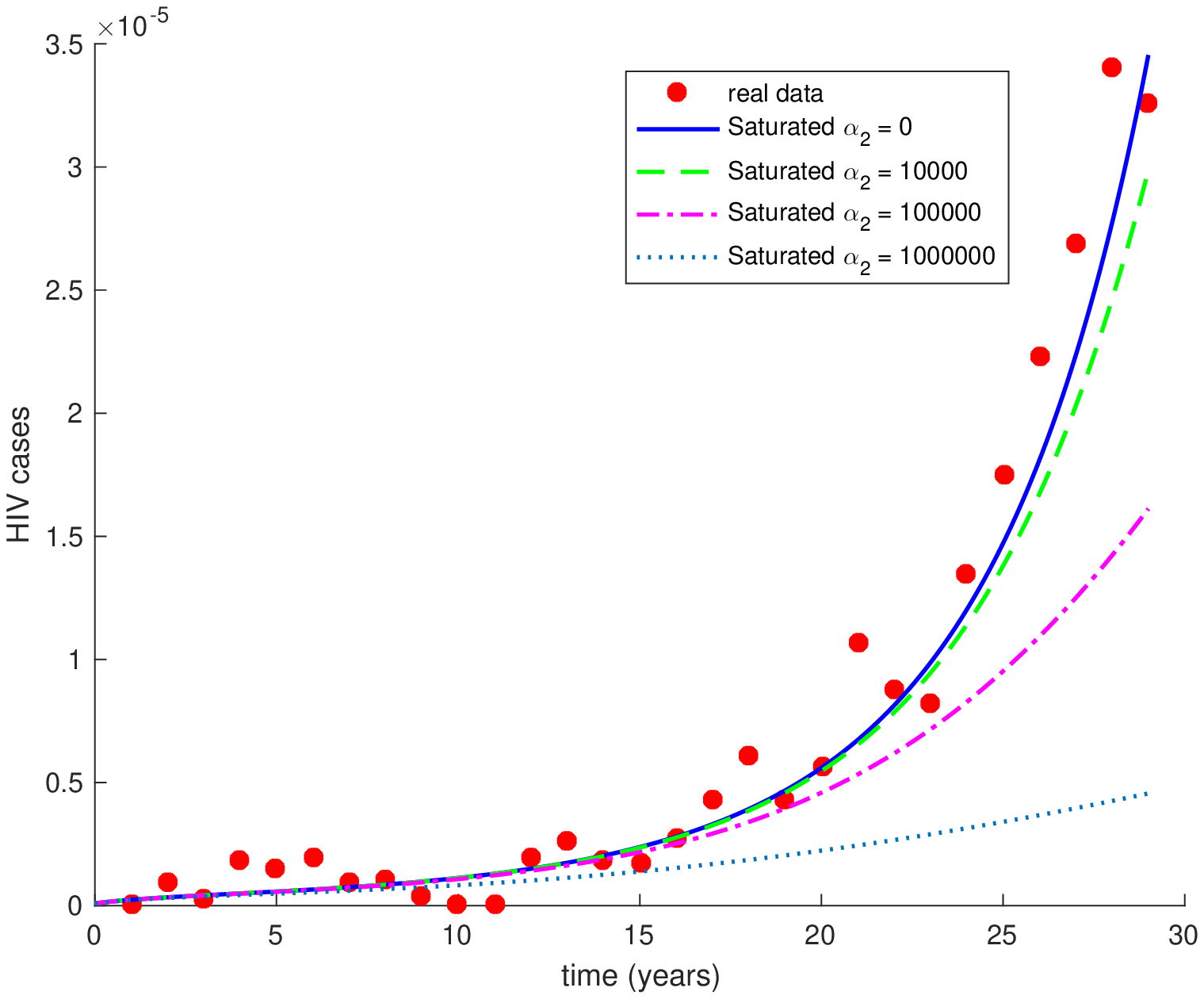}
\caption{Saturated incidence function $\dfrac{\beta S}{1+\alpha_2 I}$
with $\alpha_2 \in \{0, 10^4, 10^5, 10^6  \}$.
Time $t = 0$ corresponds to the year 1986.}
\label{fig:saturated:alpha2}
\end{figure}

For the Beddington--DeAngelis incidence function
$\dfrac{\beta S }{1 + \alpha_1 S + \alpha_2 I}$
\cite{Beddington,Cantrell,DeAngelis}, see Figure~\ref{fig:Beddington}.
\begin{figure}[!ht]
\advance\leftskip1.6cm
\includegraphics[width=0.7\textwidth]{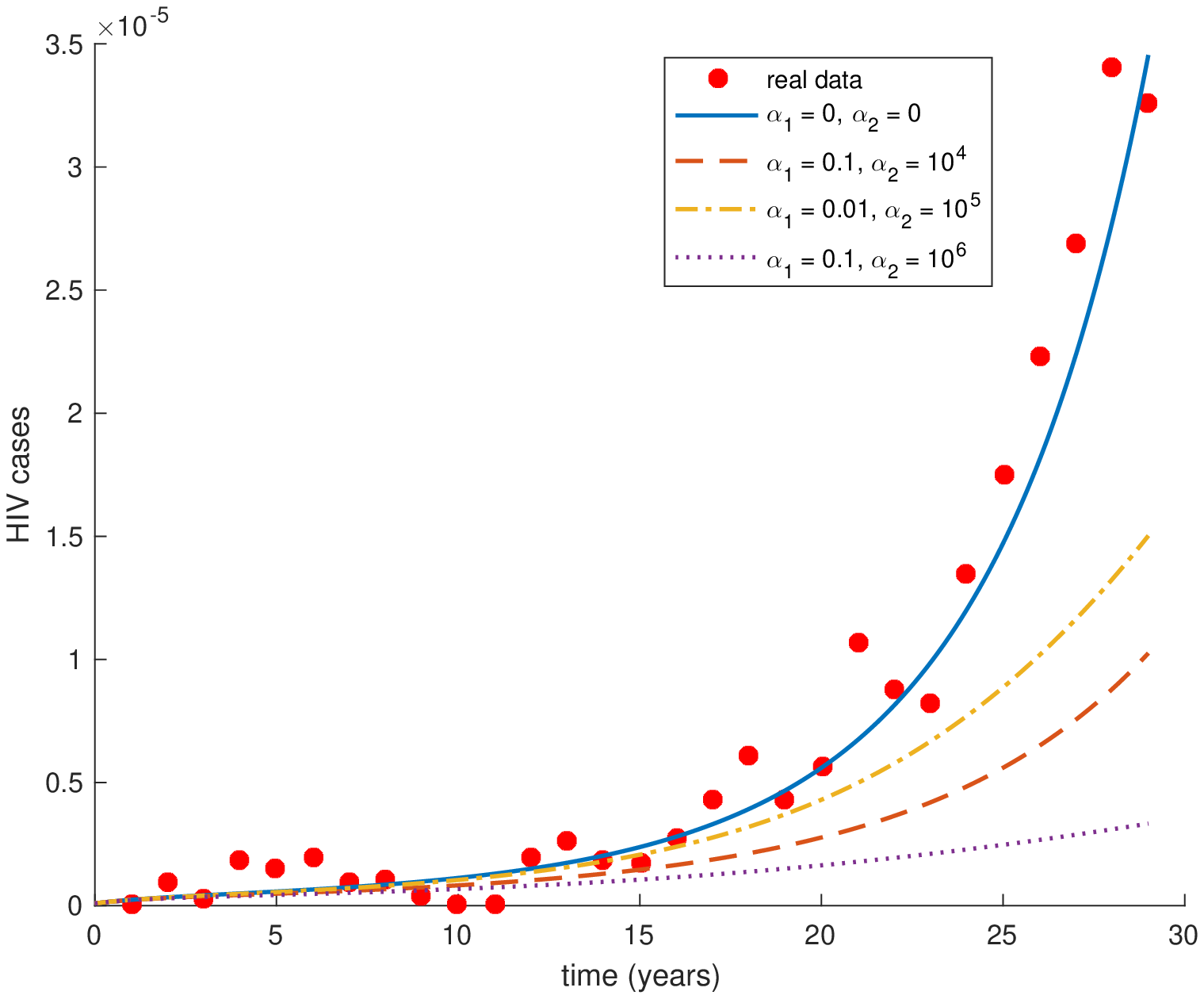}
\caption{Beddington--DeAngelis incidence function
$\dfrac{\beta S }{1 + \alpha_1 S + \alpha_2 I}$ with
$\alpha_1 \in \{0, 0.01, 0.1 \}$
and $\alpha_2 \in \{0, 10^4, 10^5, 10^6 \}$.
Time $t = 0$ corresponds to the year 1986.}
\label{fig:Beddington}
\end{figure}

For the specific non-linear incidence function
$\dfrac{\beta S}{1+\alpha_1S+\alpha_2I+\alpha_3SI}$
\cite{mahrouf1,Hattaf2,HattafD2,Lotfi,Maziane1},
see Figure~\ref{fig:specificNonlinear}.
\begin{figure}[!ht]
\advance\leftskip1.6cm
\includegraphics[width=0.7\textwidth]{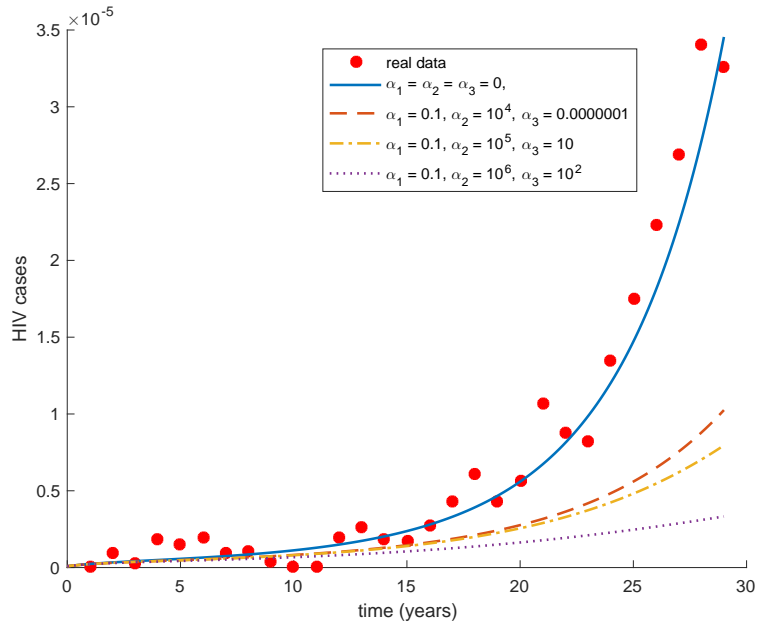}
\caption{Specific non-linear incidence function
$\dfrac{\beta S}{1+\alpha_1S+\alpha_2I+\alpha_3SI}$ with
$\alpha_1 \in \{0, 0.1 \}$, $\alpha_2 \in \{0, 10^4, 10^5, 10^6 \}$
and $\alpha_3 \in \{0, 10^{-6}, 10, 10^2 \}$.
Time $t = 0$ corresponds to the year 1986.}
\label{fig:specificNonlinear}
\end{figure}

In Tables~\ref{Table:R0} and \ref{Table:R0:alpha1},
we compute the basic reproduction number $R_0$
for each incidence function proposed in Table~\ref{Table:1}
with the parameter values from Table~\ref{values1}.
\begin{table}[!ht]
\centering
\caption{Basic reproduction number for some special incidence functions.}
\label{Table:R0}
{\renewcommand{\arraystretch}{2}	
\begin{tabular}{l c c}
\hline\hline
Incidence functions &   ${f(S,I)}$ &  $R_0$ \\
\hline\hline
Bilinear  & $\beta S$ & $7.5340$\\[0.15cm]
\hline
Saturated I & $\dfrac{\beta S}{1+\alpha_1S}$
&  $\dfrac{0.0031}{0.0009 \, \alpha_1 + 0.0004}$\\[0.15cm]
\hline
Saturated II &  $\dfrac{\beta S}{1+\alpha_2I}$ &  $7.5340$ \\[0.15cm]
\hline
Beddington--DeAngelis & $\dfrac{\beta S }{1+\alpha_1S+\alpha_2I}$
& $\dfrac{0.0031}{0.0009 \, \alpha_1 + 0.0004}$ \\[0.15cm]
\hline
Crowley--Martin & $\dfrac{\beta S}{1+\alpha_1S+\alpha_2I+\alpha_1 \alpha_2SI}$
& $\dfrac{0.0031}{0.0009 \, \alpha_1 + 0.0004}$\\[0.15cm]
\hline
Specific non-linear  & $\dfrac{\beta S}{1+\alpha_1S+\alpha_2I+\alpha_3SI}$
& $\dfrac{0.0031}{0.0009 \, \alpha_1 + 0.0004}$\\[0.15cm]
\hline
Hattaf--Yousfi & $\dfrac{\beta S}{\alpha_0+\alpha_1S+\alpha_2I+\alpha_3SI}$
& $\dfrac{0.0031}{0.0004 \, \alpha_0 + 0.0009 \, \alpha_1}$ \\[0.15cm]
\hline\hline
\end{tabular}}
\end{table}
\begin{table}[!ht]
\centering
\caption{Basic reproduction number for different values of $\alpha_1$
for the incidence function Saturated I, Beddington--DeAngelis,
Crowley--Martin and Specific non-linear.}
\label{Table:R0:alpha1}
{\renewcommand{\arraystretch}{2}	
\begin{tabular}{l c c c}
\hline\hline
$\alpha_1$ &   $0.01$ &  $0.1$ & $0.2$ \\
\hline
$R_0$ & $7.3725$ &  $6.1804$ & $5.2392$ \\ [0.15cm]
\hline
\hline
\end{tabular}}
\end{table}

In Table~\ref{table:sensitivity:index}, we present the sensitivity index
of parameters $\beta$, $\phi$, $\rho$, $\alpha$ and $\omega$,
computed for the parameter values given in Table~\ref{values1}.
\begin{table}[!ht]
\centering
\caption{Sensitivity index of $R_0$ for parameter values
given in Table~\ref{values1} for the bilinear incidence
function $f(S, I) = \beta S$.}
\label{table:sensitivity:index}
{\renewcommand{\arraystretch}{2}
\begin{tabular}{l l }
\hline \hline
{\small{Parameter}} & {\small{Sensitivity index}}\\
\hline
{\small{$\beta$}} & {\small{+1}} \\
{\small{$\phi$}} & {\small{$-0.5947$}}  \\
{\small{$\rho$}} & {\small{$-0.3437$}}\\
{\small{$\alpha$}} &  {\small{$+0.0844$}} \\
{\small{$\omega$}} &  {\small{$+0.5170$}} \\
\hline \hline
\end{tabular}}
\end{table}


\section{Conclusion}
\label{sec:06}

In this paper we considered a varying population size in a homogeneously
mixed population described by a non-linear mathematical model for HIV/AIDS
transmission. We have studied the global dynamics of the model with a general
incidence function. This general incidence function depends on the
susceptible and HIV-infected individuals with no AIDS symptoms, which are
not under ART treatment and therefore transmit the HIV virus. We have shown that
the model can describe very well the reality given by the data of HIV/AIDS
infection in Morocco from 1986 to 2015, when the incidence function
is carefully chosen. The sensitivity index results of the basic reproduction
number $R_0$ with respect to the parameters show that the effective contact
rate $\beta$ is the most sensitive parameter, that is, a small perturbation
in the effective contact rate value leads to relevant quantitative changes,
since an increase (decrease) of $\beta$ by a given percentage increases
(decreases) $R_0$ by that same percentage. For this reason, the estimation
of $\beta$ should be carefully done and targeted by intervention strategies
to engage the UNAIDS worldwide goal of ending the AIDS epidemic by 2030.


\section*{Acknowledgments}

The authors would like to express their gratitude to Dr. Karkouri,
Professor at the Faculty of Medicine and Pharmacy,
Ibn Rochd University Hospital, Morocco, for her help
to access the real data from Morocco.
Silva and Torres were supported by FCT through CIDMA,
project UID/MAT/04106/2019, and TOCCATA,
project PTDC/EEI-AUT/2933/2014,
funded by FEDER and COMPETE 2020.
Silva is also grateful to the support of FCT
through the post-doc fellowship SFRH/BPD/72061/2010.
The authors would like to thank also two
anonymous reviewers, for their critical remarks and
precious suggestions, which helped them to improve
the quality and clarity of the manuscript.


\bigskip



\end{document}